\documentclass{ejm}

\usepackage{graphicx}
\usepackage{amsmath}
\usepackage{epstopdf}
\usepackage{amssymb}
\usepackage{float}
\usepackage{xcolor}

\usepackage{ulem,color}

\newtheorem{theo}{Theorem}

\newtheorem{prop}[theo]{Proposition}

\newcommand{\R}{\mathbb{R}}
\newcommand{\de}{\hbox{d}}

\def\barh{\bar{h}}

\def\Pilr{\Pi_{\rm LR}}
\def\Pisr{\Pi_{\rm SR}}

\def\etal{{\it et al.\ }}
\def\ie{{\it i.e.\ }}

\title[Patterns formed in a thin film]{\bf Patterns formed in a thin
  film with spatially homogeneous and non-homogeneous Derjaguin
  disjoining pressure}

\author[A. S. Alshaikhi et al.]{%
  A.\ns S.\ns A\ls L\ls S\ls H\ls A\ls I\ls K\ls H\ls I$\,^1$, \ns
  M.\ns G\ls R\ls I\ls N\ls F\ls E\ls L\ls D$\,^{1,2}$\ns
  \and S.\ns K.\ns W\ls I\ls L\ls S\ls O\ls N$\,^1$
}  

\affiliation{%
  $\,^1$ Department of Mathematics and Statistics, University of
  Strathclyde, Livingstone Tower,\\
  26 Richmond Street, Glasgow G1 1XH,UK
  $\,^2$ email\textup{\nocorr: \texttt{m.grinfeld@strath.ac.uk}}
}

\date{20 December 2020; resubmitted 11 June 2021}
\pubyear{2000}
\volume{000}

\pagerange{\pageref{firstpage}--\pageref{lastpage}}

\begin{document}
\maketitle

\begin{abstract}
  We consider patterns formed in a two-dimensional thin
  film on a planar substrate with a Derjaguin disjoining pressure and
  periodic wettability stripes.
  We rigorously clarify some of the results obtained numerically by
  Honisch \etal [{\it Langmuir} 31: 10618--10631, 2015] and embed them
  in the general theory of thin-film equations.
  For the case of constant wettability, we elucidate the change in the
  global structure of branches of steady state solutions as the average
  film thickness and the surface tension are varied.
Specifically we find, by using methods of local bifurcation theory and
the continuation software package AUTO, both nucleation and metastable
regimes.
We discuss admissible forms of spatially non-homogeneous disjoining
pressure, arguing for a form that differs from the one used by Honisch
{\it et al.}, and study the dependence of the steady state solutions on the
wettability contrast in that case.
\end{abstract}


\begin{subjclass}[2020]
74K35 (Primary); 35B32 (Secondary) 
\end{subjclass}

\section{Introduction}
\label{section1}

Thin liquid films on solid substrates occur in many natural
situations. For example, they appear in tear films in the eye which
protect the cornea \cite{Braun2015} or in streams of lava from a
volcanic eruption \cite{Huppert1982}. Moreover, thin liquid films
occur in many technological applications, such as coatings
\cite{Kistler1997} and lubricants (e.g. oil films which lubricate the
piston in a car engine \cite{Wang2018}, drying paint layers
\cite{Howison1997}, and in the manufacture of microelectronic devices
\cite{Karnauschenko2020}).  For extensive reviews of thin-film flow
see, for example, Oron \etal \cite{Oron1997} and Craster and Matar
\cite{Craster2009}.

As these liquid films are thin, the Navier--Stokes equation governing
their flow can be reduced to a single degenerate fourth-order
quasi-linear parabolic partial differential equation (PDE) usually
known as a thin-film equation. In many applications a choice of a
disjoining pressure, which we denote by $\Pi$, must be made. Such a
term describes the action of surface forces on the film
\cite{Starov}. In different situations, different forms of disjoining
pressure are appropriate; these may incorporate long-range van der
Waals forces and/or various types of short-range interaction terms
such as Born repulsion; inclusion of a particular type of interaction
can have significant effects on the wettability of the surface and the
evolution of the film film, sometimes leading to dewetting phenomena,
{\it i.e.} the rupture of the film and the appearance of dry spots. (Here
and subsequently by ``wettability'' of the surface we mean the ability
of a solid surface to reduce the surface tension of a liquid on
contact with it such that it spreads over the surface and wets it.)

Witelski and Bernoff \cite{Witelski2000} were one of the first authors
to analyse mathematically the rupture of three-dimensional thin
films. In particular, considering a disjoining pressure of the form
$\Pi = -1/(3 h^3)$ (we use the sign convention adopted in Honisch
\etal \cite{Honisch2015}), they analysed planar and axisymmetric
equilibrium solutions on a finite domain.  They showed that a
disjoining pressure of this form leads to finite-time rupture
singularities, that is, the film thickness approaches zero in finite
time at a point (or a line or a ring) in the domain. In a related more
recent paper, Ji and Witelski \cite{Ji2017} considered a different
choice of disjoining pressure and investigated the finite-time rupture
solutions in a model of thin film of liquid with evaporative
effects. They observed different types of finite-time singularities due
to the non-conservative terms in the model. In particular, they showed
that the inclusion of non-conservative term can prevent the disjoining
pressure from causing finite-time rupture.

A pioneering theoretical study of a thin-film equation with a
disjoining pressure term given by a combination of negative powers of
the thin film thickness is that by Bertozzi \etal
\cite{Bertozzi2001}, who studied the formation of dewetting patterns
and the rupture of thin liquid films due to long-range attractive and
short-range Born repulsive forces, and determined the structure of the
bifurcation diagram for steady state solutions, both with and without
the repulsive term.

Aiming to quantify the temporal coarsening in a thin film, Glasner and
Witelski \cite{Glasner2005} examined two coarsening mechanisms that
arise in dewetting films: mass exchange that influences the breakdown
of individual droplets and spatial motion that results in droplet
rupture as well as merging events. They provided a simple model with a
disjoining pressure which combines the effects of both short- and
long-range forces acting on the film. Kitavtsev \etal
\cite{Kitavtsev2011} analysed the long-time dynamics of dewetting in a
thin-film equation by using a disjoining pressure similar to that used
by Bertozzi \etal \cite{Bertozzi2001}. They applied centre manifold
theory to derive and analysed an ordinary differentual equation model
for the dynamics of dewetting.

The recent article by Witelski \cite{witelski2020nonlinear} presents a
review of the various stages of dewetting for a film of liquid
spreading on a hydrophobic substrate. Different types of behaviour of
the film are observed depending on the form of the disjoining
pressure: finite-time singularities, self-similar solutions and
coarsening.  In particular, he divides the evolution of dewetting
processes into three phases: an initial linear instability that leads
to finite-time rupture (short time dynamics), which is followed by the
propagation of dewetting rims and instabilities of liquid ridges
(intermediate time dynamics), and the eventual formation of
quasi-steady droplets (long time dynamics).

Most of the previous studies of thin liquid films focussed on films on
homogeneous substrates. However, thin liquid films on non-homogeneous
chemically patterned substrates are also of interest. These have
many practical applications, such as in the construction of
microfluidic devices and creating soft materials with a particular
pattern \cite{Quake2000}. Chemically patterned substrates are an
efficient way to obtain microstructures of different shapes by using
different types of substrate patterning \cite{Sehgal2002}. Chemical
modification of substrates can also be used to avoid spontaneous
breakup of thin films, which is often highly undesirable, as, for
example, in printing technology \cite{Brasjen2011,Ondarcuhu2013}.

Due to their many applications briefly described above, films on
non-homogeneous substrates have been the object of a number of
previous theoretical studies which motivate the present work. For
example, Konnur \etal \cite{Konnur2000} found that in the case of an
isolated circular patch with wetting properties different from that of
the rest of the substrate that chemical non-homogeneity of the
substrate can greatly accelerate the growth of surface instabilities.
Sharma \etal \cite{Sharma2003} studied instabilities of a liquid film
on a substrate containing a single heterogeneous patch and a substrate
with stripes of alternating less and more wettable regions. The main
concern of that paper was to investigate how substrate patterns are
reproduced in the liquid film, and to determine the best conditions for
templating.

Thiele \etal \cite{Thiele2003} performed a bifurcation analysis using
the continuation software package AUTO \cite{AUTO2009} to study
dewetting on a chemically patterned substrate by solving a thin-film
equation with a disjoining pressure, using the wettability contrast as
a control parameter. The wettability contrast measures the degree of
heterogeneity of the substrate; it is introduced and defined
rigorously in (\ref{LDP}) in Section~\ref{section6}.  Honisch \etal
\cite{Honisch2015} modelled an experiment in which organic molecules were
deposited on a chemically non-homogeneous silicon oxide substrates
with gold stripes and discuss the redistribution of the liquid following
deposition.

In a recent paper, Liu and Witelski \cite{LiuandWitelski2020}
studied thin films on chemically heterogeneous substrates. They claim
that in some applications such as digital microfluidics, substrates
with alternate hydrophilic and hydrophobic rectangular areas are
better described by a piecewise constant function than by a sinusoidal
one. Therefore, in contrast to other studies, including the present
one, they study substrates with wettability characteristics described
by such a function. Based on the structure of the
bifurcation diagram, they divide the steady-state solutions into six
distinct but connected branches and show that the only unstable branch
corresponds to confined droplets, while the rest of the branches are stable.

In the present work, we build on the work of Thiele \etal
\cite{Thiele2003} and Honisch \etal \cite{Honisch2015}. Part of our
motivation is to clarify and explain rigorously some of the numerical
results reported in these papers. In the sinusoidally striped
non-homogeneous substrate case, we offer a justification for using a
form of the disjoining pressure that differs from the one used in
these two papers. A detailed plan of the paper is given in the last
paragraph of Section~\ref{section2}.

\section{Problem Statement}
\label{section2}


Denoting the thickness of the thin liquid film by $z=h(x,y,t)$, where
$(x,y,z)$ are the usual Cartesian coordinates and $t$ is time, 
Honisch \etal \cite{Honisch2015} consider the thin-film equation
 \begin{equation}\label{EQ1}
h_t = \nabla \cdot \left\{ Q(h) \nabla P(h,x,y) \right\},\qquad t >0,
\qquad (x,y) \in \mathbb{R}^2,
\end{equation}
where $Q(h)= h^3/(3\eta)$ is the mobility coefficient with $\eta$
being the dynamic viscosity. The 
generalized pressure $P(h,x,y)$ is given by  
\begin{equation}\label{gpress}
P(h,x,y) = -\gamma \Delta h - \Pi(h,x,y), 
\end{equation}
where $\gamma$ is the coefficient of surface tension and we follow \cite{Honisch2015}
in taking the Derjaguin disjoining pressure $\Pi(h,x,y)$ in the
spatially homogeneous case to be of the form
\begin{equation} \label{pidef1}
 \Pi(h,x,y)  = -\frac{A}{h^3} +\frac{B}{h^6} 
\end{equation}
suggested, for example, by Pismen \cite{Pismen}. Here $A$ and $B$ are
positive parameters that measure the relative contributions of the
long-range forces (the $1/h^3$ term) and the short-range ones (the
$1/h^6$ term). However, we will see that both of these constants can
be scaled out of the mathematical problem. Equation (\ref{pidef1})
uses the exponent $-3$ for the long-range forces and $-6$ for the
short-range forces as in Honisch \etal \cite{Honisch2015}.  Other
choices include the pairs of long- and short-range exponents
$(-2,-3)$, $(-3,-4)$ and $(-3,-9)$ discussed by \cite{Bertozzi2001,
  Schwartz98, Sibley}. In terms of the classification of Bertozzi
\etal \cite[Definition 1]{Bertozzi2001}, the choice $(-3,-6)$ is
\textit{admissible} (as are all the other pairs above), and falls in
their region II; we expect that choosing other admissible exponent
pairs will give qualitatively the same results as those obtained here.

In what follows, we study thin films on both homogeneous and
non-homogeneous substrates. In the non-homogeneous case, we will
modify \eqref{pidef1} by assuming that the Derjaguin pressure term
$\Pi$ changes periodically in the $x$-direction with period $L$.  The
appropriate forms of $\Pi$ in the non-homogeneous case are discussed
in Section \ref{section6}.

Hence, in order better to understand solutions of (\ref{EQ1}), we
  study its one-dimensional version,
\begin{equation}\label{qua}
  h_t = (Q(h)P(h,x)_x)_x, \quad\quad 0< x<L. 
\end{equation}

We start by characterising steady state solutions of
(\ref{qua}) subject to periodic boundary conditions at $x=0$ and
$x=L$.
In other words, we seek steady state solutions $h(x)$ of (\ref{qua}),
satisfying the non-local boundary value
problem
\begin{equation}\label{nldim}
  \gamma h_{xx} + \frac{B}{h^6} - \frac{A}{h^3} -\frac{1}{L} \int_0^L
\left[ \frac{B}{h^6} - \frac{A}{h^3}\right] \, \de{x} = 0,\; \; 0< x <L,
\end{equation}
subject to the constraint
\begin{equation}\label{nldim1}
  \frac{1}{L}\int_0^L h(x) \, \de{x} = h^*,  
\end{equation}
where the constant $h^*\,(>0)$ denotes the (scaled) average film
thickness, and the periodic boundary conditions
\begin{equation}\label{nldim2}
  h(0)=h(L),  \quad h_x(0)=h_x(L).
\end{equation}  
    
Now we non-dimensionalise. Setting
\begin{equation}
  H = \left( \frac{B}{A} \right)^{1/3}, \; \; h= H \tilde{h}, \quad
  \hbox{and} \quad x = L\tilde{x},
\end{equation}
in (\ref{nldim}) and removing the tildes, we obtain  
\begin{equation}\label{nl}
\epsilon^2 h_{xx} + f(h) - \int_0^1 f(h) \, \de{x} = 0,\; \; 0< x <1,
\end{equation}
where
\begin{equation}\label{nl1}
f(h)= \frac{1}{h^6}-\frac{1}{h^3},
\end{equation}
and
\begin{equation}\label{eps}
  \epsilon^2 = \frac{\gamma B^{4/3}}{L^2 A^{7/3}},
\end{equation}  
subject to the periodic boundary conditions
\begin{equation}\label{nl2}
h(0)=h(1), \quad h_x(0)=h_x(1),
\end{equation}
and the volume constraint
\begin{equation}\label{vc}
  \int_0^1 h(x) \, \de{x} = \barh,
\end{equation}
where
\begin{equation}\label{mc}
\barh= \frac{h^*A^{1/3}}{B^{1/3}}.
\end{equation}

Note that the problem (\ref{nl})--(\ref{mc}) is very similar to the
corresponding steady state problem for the Cahn--Hilliard equation
considered as a bifurcation problem in the parameters $\barh$ and
$\epsilon$ by Eilbeck \etal \cite{Eilbeck1989}. The boundary
conditions considered in that work were the physically natural double
Neumann conditions. The periodic boundary conditions (\ref{nl2}) in
the present problem slightly change the analysis, but our general
approach in characterising different bifurcation regimes still follows
that of Eilbeck \etal \cite{Eilbeck1989}, though the correct
interpretation of the limit as $\epsilon \to 0^+$ is that now we let
the surface tension $\gamma$ go to zero.
In particular, we perform a Liapunov--Schmidt reduction to determine
the local behaviour close to bifurcation points and then use AUTO (in
the present work we use the AUTO-07p version \cite{AUTO2009}) to
explore the global structure of branches of steady state solutions both
for the spatially homogeneous case and for the spatially
non-homogeneous case in the case of an $x$-periodically patterned
substrate.

We first investigate the homogeneous case and, having elucidated the
structure of the bifurcations of non-trivial solutions from the
constant solution $h=\barh$ in that case in Sections~\ref{section3}
and~\ref{section4}, we study forced rotational ($O(2)$) symmetry
breaking in the non-homogeneous case in Section~\ref{section6}.
In Appendix \ref{sb}, we present a general result about $O(2)$ symmetry
breaking in the spatially non-homogeneous case. It shows that in the
spatially non-homogeneous case, only two steady state solutions remain
from the orbit of solutions of (\ref{nl})--(\ref{mc}) induced by its
$O(2)$ invariance.
We concentrate on the simplest steady state solutions of
(\ref{nl})--(\ref{mc}), as by a result of Laugesen and Pugh
\cite[Theorem 1]{LP2000a} only such solutions, that is, constant
solutions and those having only one extremum point, are linearly
stable in the homogeneous case. For information about dynamics of
  one-dimensional thin-film equations the reader should also consult Zhang
  \cite{Zhang09}.

\medskip

In what follows, we use $\|\cdot\|_2$ to denote $L^2([0,1])$ norms.

\section{Liapunov--Schmidt Reduction in the Spatially Homogeneous Case}
\label{section3}

We start by performing an analysis of the dependence of the global structure of branches of steady state solutions of the problem in the spatially homogeneous case given by (\ref{nl})--(\ref{mc}) on the parameters $\barh$ and $\epsilon$.

In what follows, we do not indicate explicitly the dependence of the
operators on the parameters $\barh$ and $\epsilon$, and all of the
calculations are performed for a fixed value of $\barh$ and close to a
bifurcation point $\epsilon=\epsilon_k$ for $k=1,2,3,\ldots$ defined
below.

We set $v=h-\barh$, so that $v=v(x)$ has zero mean, and rewrite (\ref{nl}) as
\begin{equation}
G(v) = 0,
\end{equation}
where
\begin{equation}
G(v)=\epsilon^2 v_{xx} + f(v+\barh) - \int_{0}^{1}f(v(x)+\barh) \, \de{x}.
\end{equation}
If we set
\begin{equation}
H = \left\{w\in C(0,1) \,:\, \int_{0}^{1} w(x) \, \de{x} =0\right\},
\end{equation}
where
$G$ is an operator from $D(G)\subset H \rightarrow H$, then
$D(G)$ is given by
\begin{equation}
D(G) = \left\{v\in C^2(0,1) \,:\, v(0)=v(1), \, v_{x}(0)=v_{x}(1), \, \int_{0}^{1} v(x) \, \de{x} =0 \right\}.
\end{equation}
The linearisation of $G$ at $v$ applied to $w$ is defined by
\begin{equation}
dG(v) w = \lim_{\tau\to0} \frac{G(v+\tau w)-G(v)}{\tau}.
\end{equation}

We denote $dG(0)$ by $S$, so that $S$ applied to $w$ is given by
\begin{equation}\label{eqS}
S w = \epsilon^2 w_{xx} + f'(\barh)w.
\end{equation}
To locate the bifurcation points, we have to find the nontrivial
solutions  of the equation $S w = 0$ subject to
\begin{equation}\label{bcS}
w(0)=w(1), \quad \quad w_x(0)=w_x(1). 
\end{equation}
The kernel of $S$ is non-empty and two dimensional when
\begin{equation}\label{ek}
\epsilon = \epsilon_{k} = \frac{\sqrt{f'(\barh)}}{2k\pi}
\quad \textrm{for} \quad
k = 1,2,3,\ldots,
\end{equation}
and is spanned by $\cos(2k\pi x)$ and $\sin(2k\pi x)$.
That these values of $\epsilon$ correspond to bifurcation points
follows from two theorems of Vanderbauwhede \cite[Theorems 2 and 3]{Vanderbauwhede1981}.

In a neighbourhood of a bifurcation point $(\epsilon_k,0)$ in
$(\epsilon,v)$ space, solutions of $G(v)=0$ on $H$ are in one-to-one
correspondence with solutions of the reduced system of equations on
$\R^2$,
\begin{equation}\label{bf}
g_1(x,y,\epsilon)=0, \quad
g_2(x,y,\epsilon)=0,
\end{equation}
for some functions $g_1$ and $g_2$ to be obtained through the
Liapunov--Schmidt reduction \cite{Golubitsky1985}.

To set up the Liapunov--Schmidt reduction, we decompose $D(G)$ and $H$
as follows:
\begin{equation}
D(G)= \hbox{ker}\,S \oplus M
\end{equation}
and
\begin{equation}
H = N \oplus \hbox{range} \, S.
\end{equation}
Since $S$ is self-adjoint with respect to the $L^2$-inner product
denoted by $\langle \cdot, \cdot \rangle$, we can choose
\begin{equation}
M = N = \hbox{span}\, \left\{ \cos(2kx), \sin(2kx) \right\},
\end{equation}
and denote the above basis for $M$ by $\left\{ w_1, w_2 \right\}$ and for $N$ by $\left\{ w_1^*, w_2^* \right\}$.
We also denote the projection of $H$ onto $\hbox{range}\, S$ by $E$.

Since the present problem is invariant with respect to the group $O(2)$, the functions $g_1$ and $g_2$ must have the form
\begin{equation}\label{g1g2}
g_1(x,y,\epsilon) = xp(x^2+y^2,\epsilon), \quad
g_2(x,y,\epsilon) = yp(x^2+y^2,\epsilon),
\end{equation}
for some function $p(\cdot,\cdot)$ \cite{Chossat2000}, which means
that in order to determine the bifurcation structure, the only terms
that need to be computed are $g_{1,x\epsilon}$ and $g_{1,xxx}$, as
these immediately give $g_{2,y\epsilon}$ and $g_{2,yyy}$ and all of
the other second and third partial derivatives of $g_1$ and $g_2$ are
identically zero.

Following Golubitsky and Schaeffer \cite{Golubitsky1985}, we have
\begin{eqnarray}
g_{1,x\epsilon}
& = & \langle w_1^*, dG_{\epsilon}(w_1)-d^2G(w_1, S^{-1}E G_{\epsilon}(0)) \rangle, \label{222}\\
g_{1,xxx}
& = &\langle w_1^*, d^3G(w_1,w_1,w_1) - 3d^2G(w_1,S^{-1}E[d^2G(w_1,w_1)]) \rangle, \label{333}
\end{eqnarray}
where
\begin{equation} \label{drG}
d^rG(z_1,\ldots,z_r)=\left. \frac{\partial^r}{\partial t_1\ldots\partial t_r}G (t_1z_1+\ldots+t_rz_r) \right\vert_{t_1=\ldots=t_r=0} \quad \textrm{for} \quad r = 1, 2, 3, \ldots,
\end{equation}
and we choose
\begin{equation}
w_1= w_1^*=\cos(2k\pi x),
\end{equation}
where $w_1 \in \hbox{ker}\ S$ and $w_1^* \in (\hbox{range}\ S)^{\perp}$.
In particular, from (\ref{drG}) we have
\begin{eqnarray}
d^2G(z_1,z_2)
& = & \left. \frac{\partial^2}{\partial t_1\partial t_2} G(t_1z_1+t_2z_2) \right\vert_{t_1=t_2=0} \nonumber\\
& = & \frac{\partial^2}{\partial t_1\partial t_2} \Big[ \epsilon_k
      (t_1z_{1,xx}+t_2z_{2,xx}) + f(t_1z_1+t_2z_2+\barh)\nonumber \\
& - & \left. \int_{0}^{1} f(t_1z_1+t_2z_2+\barh) \, \de{x}
      \Big]\right\vert_{t_1=t_2=0}\nonumber \\
  & = & f''(\barh)z_1z_2 - \int_{0}^{1}
       f''(\barh)z_1z_2 \, \de{x},
\end{eqnarray}
and so
\begin{eqnarray}
d^2G(\cos(2k\pi x),\cos(2k\pi x))
& = & f''(\barh)\cos^2(2k\pi x)- \int_{0}^{1} f''(\barh)\cos^2(2k\pi x) \, \de{x} \nonumber\\
& = & f''(\barh)\cos^2(2k\pi x)- \frac{1}{2} f''(\barh).
\end{eqnarray}

To obtain $S^{-1}E [d^2G(w_1,w_1)]$, which we denote by $R(x)$, so
that $SR=E[d^2G(w_1,w_1)]$, we use the definition of $\epsilon_k$
given in (\ref{ek}) and solve the second order ordinary
differential equation satisfied by $R(x)$,
\begin{equation}
R_{xx} + 4k^2\pi^2 R = 4k^2\pi^2 \frac{f''(\barh)}{f'(\barh)}
\cos^2(2k\pi x) - 2k^2\pi^2 \frac{f''(\barh)}{f'(\barh)},
\end{equation}
subject to
\begin{equation}
R(0)=R(1)\quad \textrm{and} \quad R_x(0)=R_x(1) 
\end{equation}
in $\hbox{ker}\, S$. We obtain 
\begin{equation}
R(x) = S^{-1}E[d^2G(w_1,w_1)] = -\frac{1}{6}\frac{f''(\barh)}{f'(\barh)} \cos(4k\pi x),
\end{equation}
and hence
\begin{eqnarray}\label{d2}
d^2G(w_1,S^{-1}E[d^2G(w_1,w_1)]) &=& d^2G\left(\cos(2k\pi x),-\frac{1}{6}\frac{f''(\barh)}{f'(\barh)} \cos(4k\pi x)\right) \nonumber\\
& = & f''(\barh)\cos(2k\pi x)\left(-\frac{1}{6}\frac{f''(\barh)}{f'(\barh)} \cos(4k\pi x)\right) \nonumber\\
&   & - \int_{0}^{1} f''(\barh)\cos(2k\pi x)\left(-\frac{1}{6}\frac{f''(\barh)}{f'(\barh)} \cos(4k\pi x)\right) \de{x} \nonumber\\
& = & - \frac{1}{6}\frac{[f''(\barh)]^2}{f'(\barh)} \cos(2k\pi x)\cos(4k\pi x).
\end{eqnarray}
In addition, from (\ref{drG}) we have
\begin{eqnarray}
d^3G(z_1,z_2,z_3)
& = & \left. \frac{\partial^3}{\partial t_1\partial t_2\partial t_3}G(t_1z_1+t_2z_2+t_3z_3)\right\vert_{t_1=t_2=t_3=0} \nonumber\\
& = & f'''(\barh)z_1z_2z_3 - \int_0^1 f'''(\barh)z_1z_2z_3 \, \de{x},
\end{eqnarray}
and therefore
\begin{eqnarray}\label{d3}
d^3G(\cos(2k\pi x),\cos(2k\pi x),\cos(2k\pi x))
& = & f'''(\barh)\cos^3(2k\pi x) - \int_{0}^{1} f'''(\barh)\cos^3(2k\pi x) \, \de{x} \nonumber\\
& = & f'''(\barh)\cos^3(2k\pi x).
\end{eqnarray}

Putting all of the information in (\ref{d2}) and (\ref{d3}) into (\ref{333}) we obtain
\begin{eqnarray}\label{g1xxx}
g_{1,xxx}
& = & \langle w_1^*, d^3G(w_1,w_1,w_1) - 3d^2G(w_1,S^{-1}E[d^2G(w_1,w_1)]) \rangle \nonumber\\
& = & \int_{0}^{1} \cos(2k\pi x)\left[f'''(\barh)\cos^3(2k\pi x)-3\left(-\frac{1}{6}\frac{[f''(\barh)]^2}{f'(\barh)} \cos(2k\pi x)\cos(4k\pi x)\right)\right]\, \de{x} \nonumber\\
& = & \frac{3}{8}f'''(\barh) + \frac{1}{8}\frac{[f''(\barh)]^2}{f'(\barh)}.
\end{eqnarray}

In addition,
$G_{\epsilon}(v)= v_{xx}$, so that $G_{\epsilon}(0)= 0$ at $v=0$, and hence we have
\begin{equation}
d^2G(w_k,S^{-1}EG_{\epsilon}(0))= 0.
\end{equation}
Furthermore,
since $dG_{\epsilon}(w)= w_{xx}$, from (\ref{222}) we obtain
\begin{eqnarray}\label{g1xe}
g_{1,x\epsilon}
& = & \langle w_1^*, dG_{\epsilon}(w_1)-d^2G(w_1, S^{-1}EG_{\epsilon}(0)) \rangle \nonumber\\
& = & \int_{0}^{1} \cos(2k\pi x) \left(-4\pi ^2 k^2 \cos(2k\pi x)\right)\, \de{x} \nonumber\\
& = & -2 k^2 \pi^2.
\end{eqnarray}

Referring to (\ref{g1g2}) and the argument following that equation,
the above analysis shows that as long as $f'(\barh)>0$ at
$\epsilon=\epsilon_k$ a circle of equilibria bifurcates from the
constant solution $h \equiv \barh$. The direction of bifurcation is
locally determined by the sign of $g_{1,xxx}$ given by
(\ref{g1xxx}). Hence, using $1/\epsilon$ as the bifurcation parameter,
the bifurcation of nontrivial equilibria is supercritical if
$g_{1,xxx}$ is negative and subcritical if it is positive.

By finding the values of $\barh$ where $g_{1,xxx}$ given by
(\ref{g1xxx}) with $f(h)$ given by (\ref{nl1}) is zero, we finally
obtain the following proposition:

\begin{prop} \label{prop1} Bifurcations of nontrivial solutions from
  the constant solution $h=\barh$ of the problem
  $(\ref{nl})$--$(\ref{mc})$ are supercritical if
  $1.289 < \barh < 1.747$ and subcritical if $1.259 < \barh <1.289$ or
  if $\barh > 1.747$.
\end{prop}

\begin{proof}
The constant solution $h \equiv \barh$ will lose
stability as $\epsilon$ is decreased only if $f'(\barh) >0$. {\it i.e.} if
$-6/{\barh}^7 + 3/{\barh}^4 > 0$, for $\barh > 2^{1/3} \approx
1.259$. From (\ref{g1xxx}) we have that
\begin{equation}
g_{1,xxx} =   \frac{57 \barh^6-426 \barh^3+651}{2\barh^9(\barh^3-2)},
\end{equation}
so that $g_{1,xxx}<0$ if $\barh \in (1.289, 1747)$ giving the
result.
\end{proof}

For $\barh \leqslant 2^{1/3}$ there are no bifurcations from the
constant solution $h = \barh$. Furthermore, we have the following
proposition:

\begin{prop}\label{prop2}
The problem $(\ref{nl})$--$(\ref{mc})$ has no nontrivial solutions when $\barh \leq 1$.
\end{prop}

\begin{proof} Assume that such a nontrivial solution exists. Then,
since $\barh \leq 1$, its global minimum, achieved at some point $x_0
\in (0,1)$, must be less than 1. (We may take the point $x_0$ to be an
interior point by translation invariance.) But then
\begin{equation}
\epsilon^2 h_{xx}(x_0) = \int_0^1 f(h)\, \de{x} - f(h(x_0)) < 0,
\end{equation}
so the point $x_0$ cannot be a minimum.
\end{proof}

\section{Two-Parameter Continuation of Solutions in the
  Spatially Homogeneous Case}
\label{section4}

To describe the change in the global structure of branches of
steady state solutions of the problem (\ref{nl})--(\ref{mc}) as $\barh$
and $\epsilon$ are varied, we use AUTO \cite{AUTO2009} and our results
are summarised in Figure \ref{fig:1}.

As Figure \ref{fig:1} shows, a curve of saddle-node (SN)
bifurcations which originates from $\barh \approx 1.289$ at
$1/\epsilon \approx 23.432$ satisfies $\barh \rightarrow
1^+$ as $1/\epsilon \rightarrow \infty$, while  a curve of SN
bifurcations which originates from $\barh \approx 1.747$, $1/\epsilon
\approx 13.998$ satisfies $\barh \rightarrow \infty$ as
$1/\epsilon \rightarrow \infty$.

\begin{figure}[H] 
\begin{center}
\includegraphics[width=0.7\linewidth]{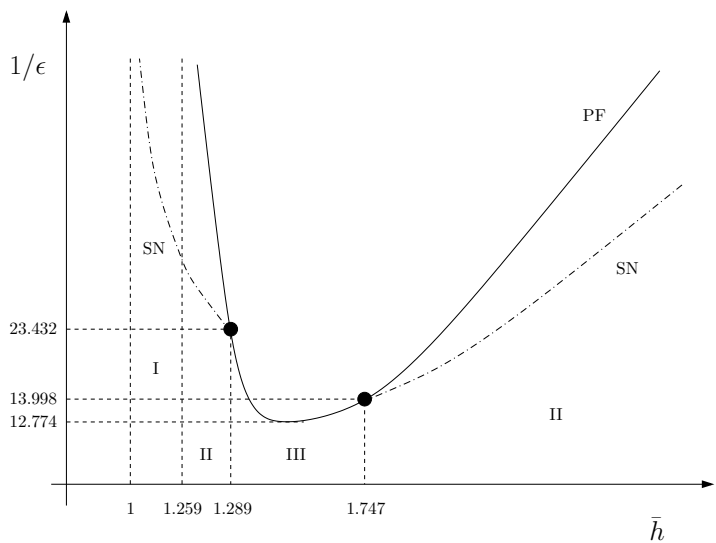}
\caption{The global structure of branches of steady state solutions
  with a unique maximum, including both saddle-node (SN) (shown with
  dash-dotted curves) and pitchfork (PF) bifurcation branches (shown
  with solid curves). The nucleation regime
  $1 < \barh < 2^{1/3} \approx 1.259$ (Regime I), the metastable
  regime $2^{1/3} < \barh < 1.289$ and $\barh > 1.747$ (Regime II),
  and the unstable regime $1.289 < \barh < 1.747$ (Regime III) are
  also indicated.}
\label{fig:1}
\end{center}
\end{figure}

Figure \ref{fig:1} identifies three different
bifurcation regimes, denoted by I, II and III, with differing
bifurcation behaviour occurring in the different regimes, namely
(using the terminology of \cite{Eilbeck1989} in the context of the
Cahn-Hilliard equation):
\begin{itemize}
\item
a ``nucleation'' regime for $1 < \barh < 2^{1/3} \approx 1.259$ (Regime I),
\item
a ``metastable'' regime for $2^{1/3} < \barh < 1.289$ and $\barh >
1.747$ (Regime II), and
\item
an ``unstable'' regime for $1.289 < \barh < 1.747$ (Regime III).
\end{itemize}

In Regime I, the constant solution $h(x)\equiv\barh$ is linearly
stable which follows from analysing the spectrum of the operator $S$
for $f'(\barh)<0$ in (\ref{eqS}) and (\ref{bcS}), but under
sufficiently large perturbations the system will evolve to a
non-constant steady state solution. See \cite{LP2000a} for an
extensive discussion of the stability analysis of steady state
solutions to thin-film equations.

In Regime II, as $\epsilon$ is decreased the constant solution $h(x)
\equiv \barh$ loses stability through a subcritical bifurcation.

In Regime III, as $\epsilon$ is decreased, the
constant solution $h(x) \equiv \barh$ loses stability through a
supercritical bifurcation.


\section{The Spatially Non-Homogeneous Case}
\label{section6}

Honisch \etal \cite{Honisch2015} chose the Derjaguin disjoining
pressure $\Pi(h,x,y)$ to be of the form
\begin{equation}\label{LDP}
\Pi(h,x,y) = \left( \frac{1}{h^6}-\frac{1}{h^3} \right) (1 + \rho\, G(x,y)),
\end{equation}
where the function $G(x,y)$ models the non-homogeneity of the
substrate and the parameter $\rho$, which can be either positive or
negative, is called the ``wettability contrast''.
Following Honisch \etal \cite{Honisch2015},
in the remainder of the present work we consider the specific case
\begin{equation}\label{G}
G(x,y) = \sin \left(2\pi x \right) := G(x),
\end{equation}
corresponding to an $x$-periodically patterned ({\it i.e.} striped)
substrate.

There are, however, some difficulties in accepting (\ref{LDP}) as a
physically realistic form of the disjoining pressure for a
non-homogeneous substrate. The problems arise because the two terms in
(\ref{LDP}) represent rather different physical effects. Specifically,
since the $1/h^6$ term models the short-range interaction amongst the
molecules of the liquid and the $1/h^3$ term models the long-range
interaction, assuming that {\it both} terms reflect the patterning in
the substrate in {\it exactly the same way} through their dependence
on the {\it same} function $G(x,y)$ does not seem very plausible.
Moreover, there are other studies which assume that the wettability of
the substrate is incorporated in either the short-range interaction
term (see, for example, Thiele and Knobloch \cite{Thiele2006} and
Bertozzi \etal \cite{Bertozzi2001}) or the long-range interaction term
(see, for example, Ajaev \etal \cite{Ajaev2011} and Sharma \etal
\cite{Sharma2003}), but not both simultaneously.
Hence in what follows we will consider the two cases
$\Pi(h,x) = \Pilr(h,x)$ and $\Pi(h,x) = \Pisr(h,x)$,
where LR stands for ``long range'' and SR stands for ``short range'' where 
\begin{equation}\label{DPlr}
  \Pilr(h,x) = \frac{1}{h^6}-\frac{1}{h^3}(1+\rho G(x))
\end{equation}
and
\begin{equation}\label{DPsr}
\Pisr(h,x) = \frac{1}{h^6}(1+\rho G(x)) -\frac{1}{h^3}, 
\end{equation}
in both of which $G(x)$ is given by (\ref{G}) and $\rho$ is the
wettability contrast.

For small wettability contrast, $\vert\rho\vert \ll 1$, we do not
expect there to be significant differences between the influence of
$\Pilr$ or $\Pisr$ on the bifurcation diagrams, as these results
depend only on the nature of the bifurcation in the
homogeneous case $\rho=0$ and on the symmetry groups under which the
equations are invariant.
To see this, consider the spatially non-homogeneous version of
(\ref{nl}), that is, the boundary value problem
\begin{equation}\label{nln}
\epsilon^2 h_{xx} + f(h,x) - \int_0^1 f(h,x) \, \de{x} = 0,\; \; 0<x<1,
\end{equation}
where now
\begin{equation}\label{nl1n}
f(h,x)=\Pilr(h,x) \quad \hbox{or} \quad f(h,x)= \Pisr(h,x). 
\end{equation}
subject to the periodic boundary conditions and the volume constraint,
\begin{equation}\label{nl2n}
h(0)=h(1), \quad h_x(0)=h_x(1), \quad \hbox{and} \quad \int_0^1 h(x) \, \de{x} = \barh.
\end{equation}

Seeking an asymptotic solution to (\ref{nln})--(\ref{nl2n}) in the
form $h(x) = \barh + \rho h_1(x) + O(\rho^2)$ in the limit
$\rho \to 0$, by substituting this anzatz into (\ref{nln}) we find
that in the case of $\Pilr(h,x)$ we have
\begin{equation}\label{h1-lr}
h_1(x) = \frac{S\barh^4\sin\left(2\pi x\right)}{4\pi^2\barh^7\epsilon^2
  -3\barh^3 + 6},
\end{equation}
while in the case of $\Pisr(h,x)$ the corresponding result is
\begin{equation}\label{h1-sr}
  h_1(x) = \frac{\barh\sin\left(2\pi x\right)}{4\pi^2\barh^7\epsilon^2
  - 3\barh^3 + 6}.
\end{equation}

For non-zero values of $\rho$, in both the $\Pilr$ and $\Pisr$ cases,
the changes in the bifurcation diagrams obtained in the homogeneous
case ($\rho=0$) are an example of forced symmetry breaking (see, for
example, Chillingworth \cite{Chillingworth1985}), which we discuss
further in Appendix \ref{sb}. More precisely, we show there that when
$\rho \neq 0$, out of the entire $O(2)$ orbit only two equilibria are
left after symmetry breaking.

\begin{figure}[H]
\begin{center}
        \includegraphics[height=0.55\textheight]{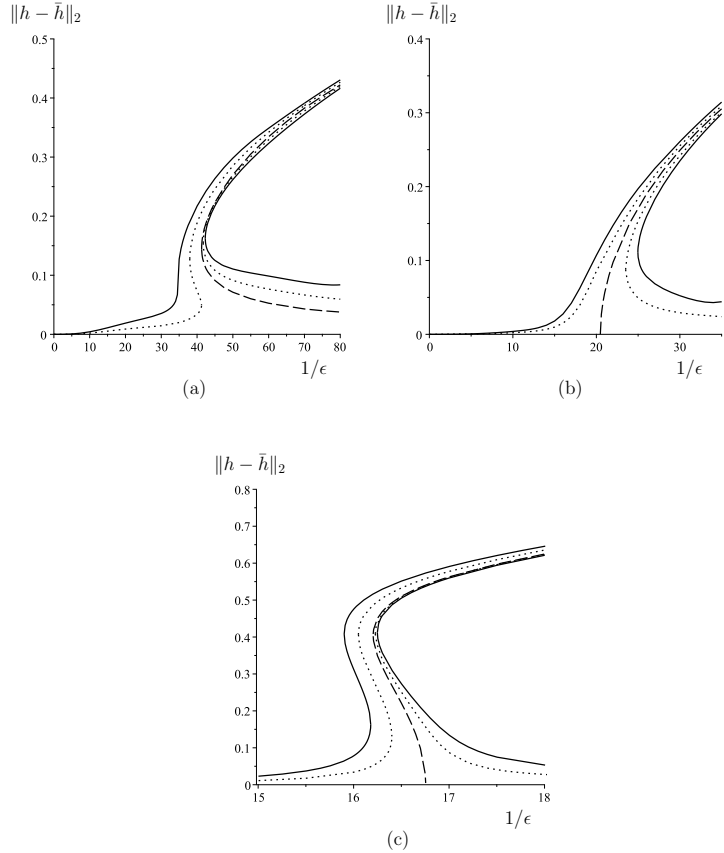}
        \caption{ Bifurcation diagrams of solutions with a unique
          maximum, showing $\|h-\barh \|_2$ plotted
          as a function of $1/\epsilon$ when the disjoining pressure
          is $\Pilr$ for $\rho=0$ (dashed curves), $\rho=0.005$ (dotted
          curves) and $\rho=0.05$ (solid curves) for (a) $\barh=1.24$,
          (b) $\barh=1.3$, and (c) $\barh=2$, corresponding to Regimes
          I, III, and II, respectively.}
        \label{fig:2}
\end{center}        
\end{figure} 
Figure \ref{fig:2} shows how for small wettability contrast
$\vert\rho\vert \ll 1$, the resulting spatial non-homogeneity
introduces imperfections \cite{Golubitsky1985} in the bifurcation
diagrams of the homogeneous case $\rho=0$ discussed in
Section~\ref{section4}.  It presents bifurcation diagrams in Regimes
I, II and III when the disjoining pressure $\Pilr$ is given by
(\ref{DPlr}) for a range of small values of $\rho$ together with the
corresponding diagrams in the case $\rho=0$. The corresponding
bifurcation diagrams when the disjoining pressure $\Pisr$ is given by
(\ref{DPsr}) are very similar and hence are not shown here.

For large wettability contrast, specifically for
$\vert\rho\vert \geq 1$, significant differences between the two forms
of the disjoining pressure are to be expected. When using $\Pilr$, one
expects global existence of positive solutions for all values of
$\vert\rho\vert$; see, for example, Wu and Zheng \cite{Wu2007}. On the
other hand, when using $\Pisr$, there is the possibility of rupture of
the liquid film for $\vert\rho\vert \geq 1$; see, for example,
\cite{Bertozzi2001,Wu2007}, which means in this case we do not expect
positive solutions for sufficiently large values of $\vert\rho\vert$.

\begin{figure}[H]
\begin{center}
\includegraphics[height=0.37\textheight]{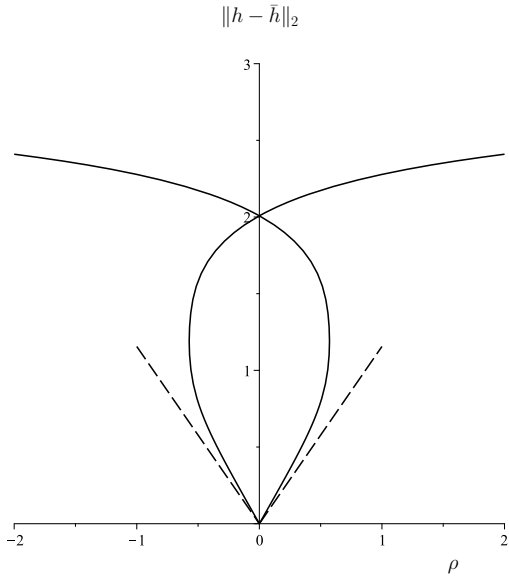}
\caption{Bifurcation diagram for steady state solutions with a unique
  maximum showing $\|h-\barh\|_2$ plotted as a function of $\rho$ when
  the disjoining pressure is $\Pilr$, $\barh=3$ and
  $1/\epsilon=50$. The leading-order dependence of $\|h-\barh\|_2$ on
  $\rho$ as $\rho \to 0$, given by (\ref{h1-lr}), is shown with dashed
  lines.}
\label{fig:3}
\end{center}
\end{figure}

In Figure \ref{fig:3} we plot the branches of the positive solutions
of (\ref{nln})--(\ref{nl2n}) with a unique maximum when the disjoining
pressure is $\Pilr$ for $\barh=3$ and $1/\epsilon=50$, so that when
$\rho=0$, we are in Regime II above the curve of pitchfork (PF)
bifurcations from the constant solution shown in  Figure~\ref{fig:1}. The
work of Bertozzi \etal \cite{Bertozzi2001} and of Wu and Zheng
\cite{Wu2007}, shows that strictly positive solutions exist for all
values of $|\rho|$, {\it i.e.} beyond the range $\rho \in [-2,2]$ for which we have 
performed the numerical calculations.

Figure \ref{fig:4} shows that the situation is different when the
disjoining pressure is $\Pisr$ (with the same values of $\barh$ and
$\epsilon$). At $\vert\rho\vert=1$, the upper branches of solutions
disappear due to rupture of the film, so that at some point
$x_0 \in [0,1]$ we have $h(x_0)=0$ and a strictly positive solution no
longer exists, while the other two branches coalesce at a saddle node
bifurcation at $\vert\rho\vert \approx 5.8$.

Note that in Figures \ref{fig:3} and \ref{fig:4}, the non-trivial
``solution'' on the axis $\rho=0$ is, in fact, a whole $O(2)$-symmetric
orbit of solutions predicted by the analysis leading to Figure~\ref{fig:1}.

\begin{figure}[H]
\begin{center}
\includegraphics[height=0.33\textheight]{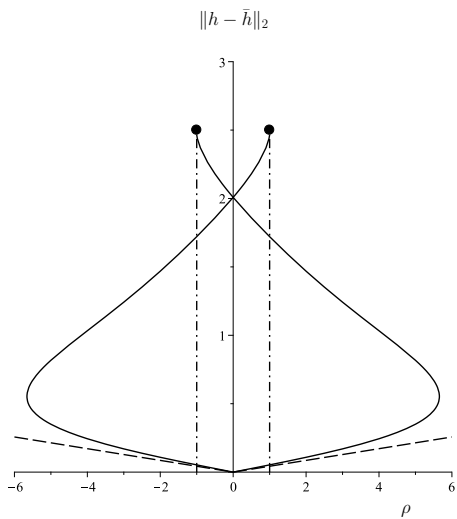}
\caption{Bifurcation diagram showing $\|h-\barh\|_2$ plotted as a
  function of $\rho$ when the disjoining pressure is $\Pisr$, for
  $\barh=3$ and $1/\epsilon=50$. The leading order dependence of
  $\|h-\barh\|_2$ on $\rho$ as $\rho \to 0$, given by
  (\ref{h1-sr}), is shown with dashed lines. Note that the upper
    branches of solutions cannot be extended beyond $|\rho|=1$
    (indicated by filled circles).}
\label{fig:4}
\end{center}
\end{figure}

Note that when the disjoining pressure is $\Pisr$, given by
(\ref{DPsr}), we were unable to use AUTO to continue branches of solutions
beyond the rupture of the film at $\vert\rho\vert=1$.
\begin{figure}[H]
\begin{center}
\includegraphics[width=0.40\textwidth]{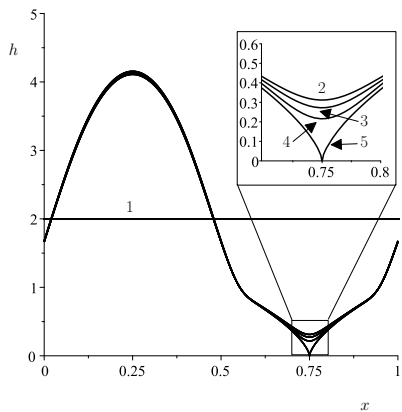}
\caption{Solutions $h(x)$ when the disjoining pressure is $\Pisr$ for
  $\barh=2$ and $1/\epsilon=30$ for $\rho=0$, 0.97, 0.98, 0.99 and 1,
  denoted by ``1'', ``2'', ``3'', ``4'' and ``5'', respectively,
  showing convergence of strictly positive solutions to a 
  non-strictly positive one as $\rho \to 1^-$.}
\label{fig-touchdown1}
\end{center}
\end{figure}

Figure~\ref{fig-touchdown1} shows the film approaching rupture as
$\rho \to 1^-$ at the point where the coefficient of the short-range
term disappears when $\rho=1$, \ie when $1 + \sin(2\pi x) = 0$ and
hence at $x=3/4$. These numerical results are consistent with the
theoretical arguments of Bertozzi \etal \cite{Bertozzi2001}.

Investigation of the possible leading-order balance in
(\ref{nl}) when the disjoining pressure is $\Pisr$ and $\rho=1$
in the limit $x \to 3/4$ reveals that $h(x) = O\left((x-3/4)^{2/3}\right)$;
continuing the analysis to higher order yields
\begin{equation}\label{nl15}
  h = (2\pi^2)^{\frac{1}{3}}\left(x-\frac{3}{4}\right)^{\frac{2}{3}}
  - \frac{4\left(4\pi^{10}\right)^{\frac13} \epsilon^2}{27} \left(x-\frac{3}{4}\right)^{\frac{4}{3}}
  + O\left(\left(x-\frac{3}{4}\right)^2\right).
\end{equation}

Figure \ref{fig-touchdown2} shows the detail of the excellent
agreement between the solution $h(x)$ when $\rho=1$ and the two-term
asymptotic solution given by (\ref{nl15}) close to 
$x=3/4$.

\begin{figure}[H]
\begin{center}
\hspace*{-0cm}
\includegraphics[width=0.45\textwidth]{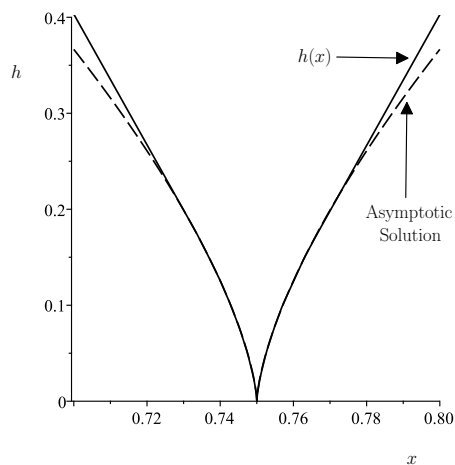}
\caption{ Detail near $x=3/4$ of the solution $h(x)$ shown with
  solid line when the disjoining pressure is $\Pisr$ and
  $\rho=1$, and the two-term asymptotic solution given by (\ref{nl15})
  shown with dashed lines for $\barh=2$ and $\epsilon=1/30$.  }
\label{fig-touchdown2}
\end{center}
\end{figure}

Figures~\ref{MOS-for-form1} and \ref{MOS-for-form2}
show the multiplicity of solutions with a unique maximum for the
disjoining pressures $\Pilr$ and $\Pisr$, respectively, in the
$(1/\epsilon,\rho)$-plane in the three regimes shown in Figure \ref{fig:1}.

\begin{figure}[H]
  \centering
  \includegraphics[width=\linewidth,height=0.8\textheight]{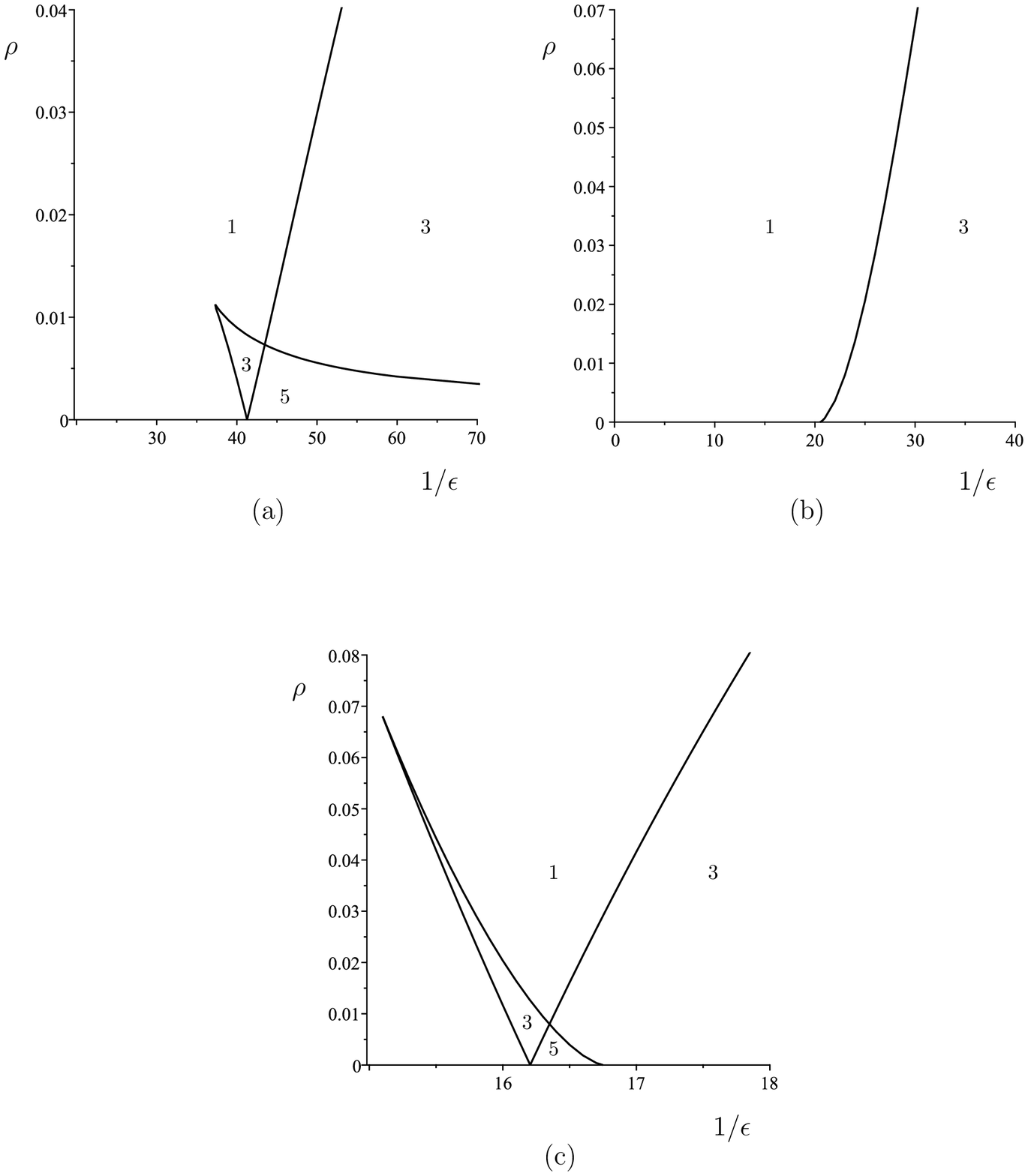}
    \caption{Multiplicity of strictly positive solutions
      with a unique maximum in the $(1/\epsilon,\rho)$-plane when the
      disjoining pressure is $\Pilr$ for (a) $\barh=1.24$ (Regime I),
      (b) $\barh=1.3$ (Regime III), and (c) $\barh=2$ (Regime
      II).
      }
 \label{MOS-for-form1}
\end{figure} 
\begin{figure}[H]
  \centering
  \includegraphics[width=\linewidth,height=0.75\textheight]{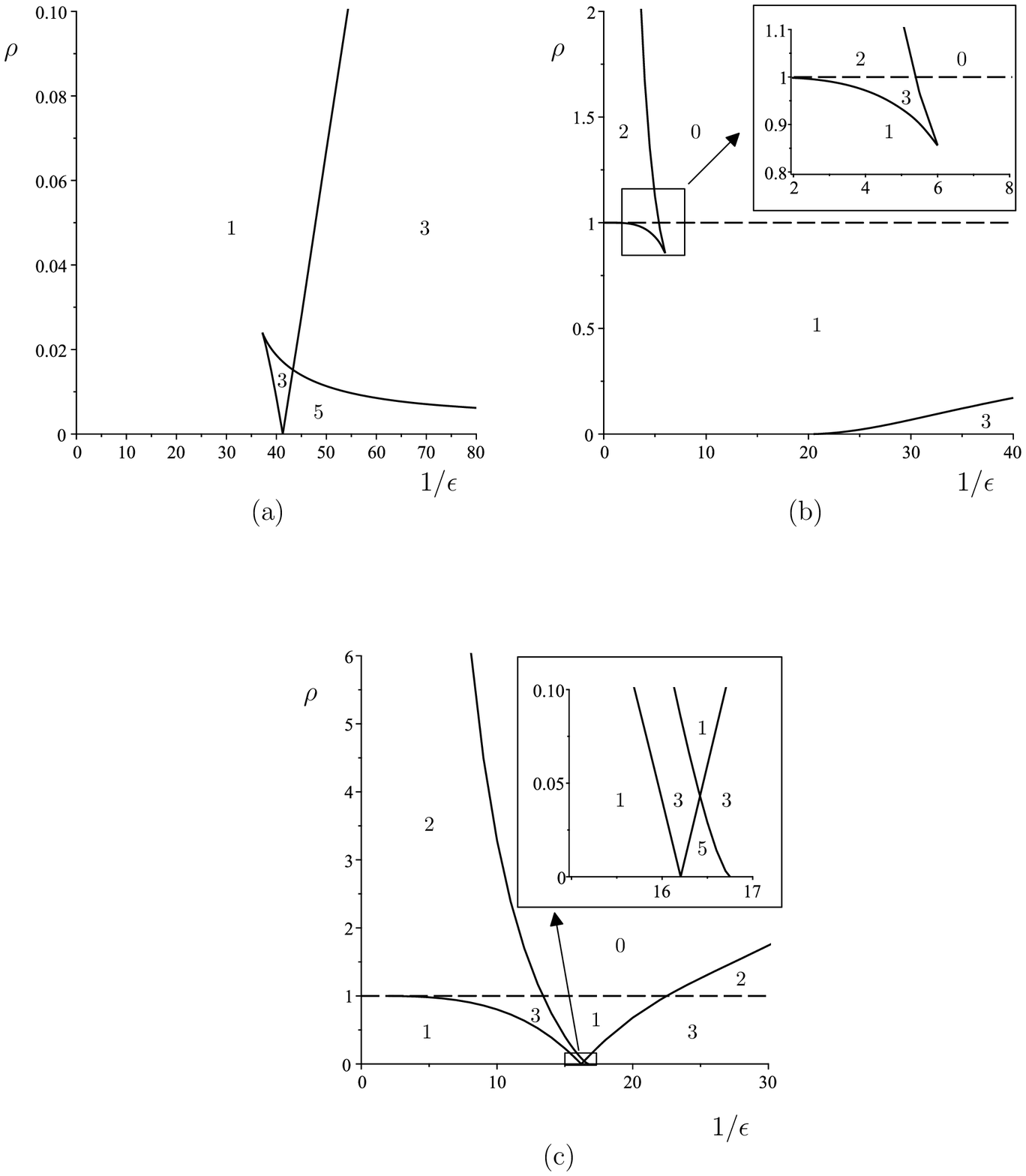}
        \caption{Multiplicity of strictly positive
          solutions with a unique maximum in the
          $(1 / \epsilon, \rho)$-plane when the disjoining pressure is
          $\Pisr$ for (a) $\bar{h}=1.24$ (Regime I), (b) $\bar{h}=1.3$
          (Regime III), and (c) $\bar{h}=2$ (Regime II).
          }
 \label{MOS-for-form2}
\end{figure} 
In Figure~\ref{MOS-for-form2} the horizontal dashed line at $\rho =1$
indicates rupture of the film and loss of a smooth strictly positive
solution, showing that there are regions in parameter space where no
such solutions exist.

Let us explain in detail the interpretation of
Figure~\ref{MOS-for-form1}(c); all of the other parts of
Figure~\ref{MOS-for-form1} and Figure~\ref{MOS-for-form2} are
interpreted in a similar way.

\begin{figure}[H]
\begin{center}
\hspace*{-0cm}
\includegraphics[width=0.45\textwidth]{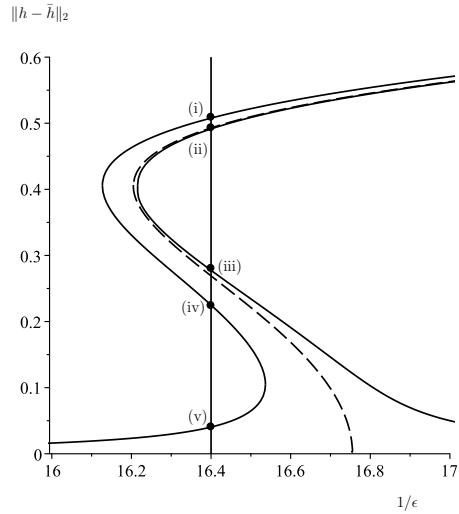}
\caption{Bifurcation diagram of steady state solutions with
  $\barh = 2$ (Regime II) for $\rho=0$ (dashed curves) and
  $\rho=0.005$ (solid curves) indicating the
  different branches of steady state solutions.}
\label{MOS-new}
\end{center}
\end{figure}
Each curve in Figure~\ref{MOS-for-form1}(c) is a curve of saddle-node
bifurcations. Similar to Figure~\ref{fig:2}(c), Figure~\ref{MOS-new}
shows the bifurcation diagram of steady state solutions with
$\barh = 2$ (Regime II) for $\rho=0$ and $\rho=0.005$. As $1/\epsilon$
is increased, we pass successively through saddle-node bifurcations
with the multiplicity of the steady state solutions changing from 1 to 3
to 5 and then back to 3 again. In Figure~\ref{5statsolns}, we plot the
five steady state solutions with a unique maximum indicated in
Figure~\ref{MOS-new} by (i)--(v).

\begin{figure}[H]
\begin{center}
\hspace*{-0cm}
\includegraphics[width=0.45\textwidth]{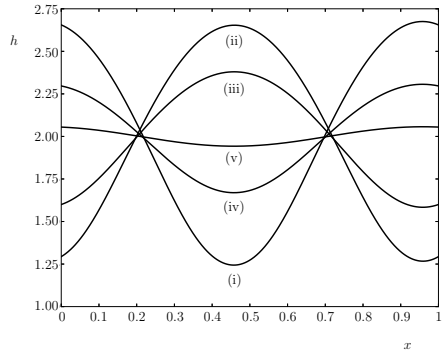}
\caption{Steady state solutions on the five branches of solutions
  indicated in Figure~\ref{MOS-new} by (i)--(v).}
\label{5statsolns}
\end{center}
\end{figure}


\section{Conclusions}
\label{section7}

In the present work we have investigated the steady state solutions of
the thin-film evolution equation \eqref{EQ1} both in the
spatially homogeneous case (\ref{nl})--(\ref{nl2}) and in the
spatially non-homogeneous case for two choices of spatially
non-homogeneous Derjaguin disjoining pressure given by (\ref{DPlr})
and (\ref{DPsr}). We have given a physical motivation for our choices
of the disjoining pressure. For reasons explained in the last
paragraph of Section~\ref{section2}, we concentrated on
branches of solutions with a unique maximum.

In the spatially homogeneous case (\ref{nl})--(\ref{mc}), we used the
Liapunov-Schmidt reduction of an equation invariant under the action
of the $O(2)$ symmetry group to obtain local bifurcation results and
to determine the dependence of the direction and nature of bifurcation
in bifurcation parameter $1/\epsilon$ on the average film thickness
$\bar{h}$; our results on the existence of three different bifurcation
regimes, (namely nucleation, metastable, and unstable) are summarised
in Propositions \ref{prop1} and \ref{prop2} and in Figure~\ref{fig:1} obtained using AUTO.

In the spatially non-homogeneous case (\ref{nln})--(\ref{nl2n}), we
clarified the $O(2)$ symmetry breaking phenomenon (see Appendix
\ref{sb}) and presented imperfect bifurcation diagrams in
Figure~\ref{fig:2} and global bifurcation diagrams using the
wettability contrast $\rho$ as a bifurcation parameter for fixed
$\epsilon$ and $\barh$ in Figures~\ref{fig:3} and \ref{fig:4}.

Let us discuss Figure~\ref{fig:3} in more detail; for convenience, it
is reproduced in Figure~\ref{fig:C1} with labelling added to the
different branches of strictly positive steady state solutions with a
unique maximum. Below we explain what these different labels mean.

\begin{figure}[H]
  \centering
  \includegraphics[height=0.35\textheight]{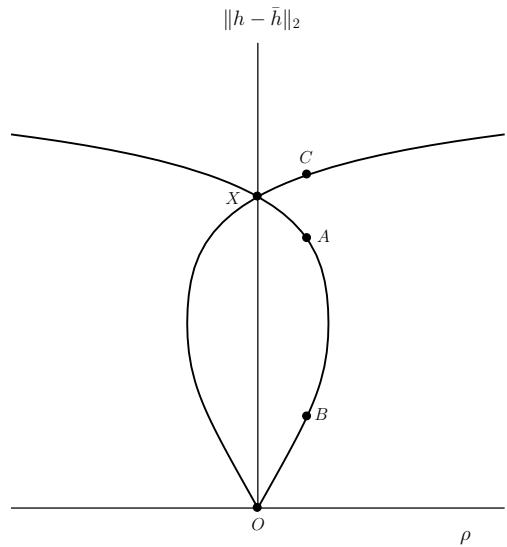}
  \caption{A sketch of the bifurcation diagram plotted in
    Figure~\ref{fig:3} with the different solution branches
    labelled.} \label{fig:C1}
\end{figure}

Our results can be described using the language of global compact
attractors. For more information on attractors in dissipative
  infinite-dimensional dynamical systems please see the monograph of
  Hale \cite{Hale10} and Wu and Zheng \cite{Wu2007} for applications
  in the thin-film context. In systems such as (\ref{qua}), the global
  compact attractor of the PDE is the union of equilibria and their
  unstable manifolds. In Figures~\ref{fig:C2} and \ref{fig:C3} we
  sketch the global attractor of (\ref{qua}) for $\epsilon=1/50$ and
  $\barh=3$, the values used to plot Figure~\ref{fig:3}. For
  these values of the parameters the attractor is two-dimensional and we
  sketch its projection onto a plane.

\begin{figure}[H]
  \centering
    \includegraphics[width=0.4\linewidth]{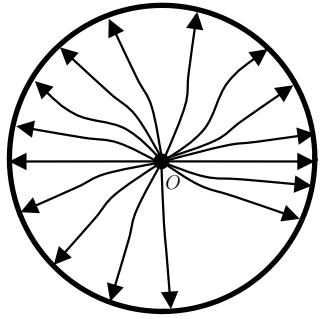}
    \caption{Sketch of the global attractor for $\rho=0$. The
      circle represents the $O(2)$ orbit of steady state solutions and $O$
      represents the constant solution $h(x)=\barh$.} \label{fig:C2}
\end{figure}
\begin{figure}[H]
  \centering
    \includegraphics[width=0.4\linewidth]{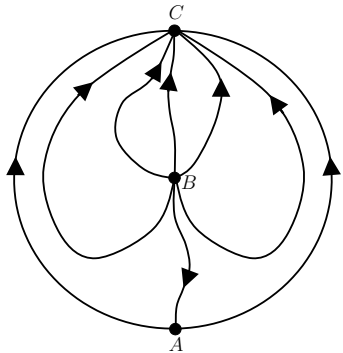}
    \caption{Sketch of the global attractor for small non-zero values
      of $|\rho|$. The points $A$, $B$, $C$ correspond to the steady
      state solutions labelled in Figure~\ref{fig:C1}.} \label{fig:C3}
\end{figure}

When $\rho=0$, for $1/\epsilon=50$, the attractor is two-dimensional;
the constant solution $h \equiv \barh$ denoted by $O$ has a
two-dimensional unstable manifold and $X$ corresponds to a whole
$O(2)$ orbit of steady state solutions.  A sketch of the attractor in
this case is shown in Figure~\ref{fig:C2}.

When $|\rho|$ takes small positive values, only two steady state
solutions, denoted by $A$ and $C$ remain from the entire $O(2)$
orbit, as discussed in Appendix \ref{sb}, while the constant
solution continues to $B$ without change of stability. The resulting
attractor is sketched in Figure~\ref{fig:C3}.

Increasing $|\rho|$ causes the steady state solutions $B$ and $C$ to
coalesce in a saddle node bifurcation, so that the
attractor degenerates to a single asymptotically stable
steady state solution. It would be interesting to understand why this collision
of steady state solution branches occurs.

We have also explored the geometry of film rupture which occurs as
$\rho \to 1^-$ when the disjoining pressure is given by $\Pisr$; this
phenomenon is shown in detail in Figures~\ref{fig-touchdown1} and
\ref{fig-touchdown2}.

Finally, in Figures \ref{MOS-for-form1} and \ref{MOS-for-form2}, we
showed the results of a two-parameter continuation study in the
$(1/\epsilon, \rho)$ plane, showing how the multiplicity of positive
steady state solutions changes as parameters are varied, and, in
particular, indicating in the case of disjoining pressure $\Pisr$ shown
in Figure~\ref{MOS-for-form2} regions in parameter space where no such
solutions exist. We conjecture that in these regions the solution of
the unsteady problem with any positive initial condition converges to
a weak solution of the thin-film equation with regions in which
$h(x)=0$, {\it i.e.} solutions with film rupture. For a
discussion of such (weak) solutions of thin-film equations in the
homogeneous case the reader is referred to the work of Laugesen and
Pugh \cite{LP2000b}.

In the case of disjoining pressure $\Pisr$, we could not use the
AUTO-07p version of AUTO to continue branches of solutions beyond
rupture. It would be an interesting project to develop such a
capability for this powerful and versatile piece of software.

Figures \ref{MOS-for-form2}(b) and \ref{MOS-for-form2}(c) provide
numerical evidence for the existence of a curve of saddle-node
bifurcations converging to the point $(0,1)$ in the
$(1/\epsilon, \rho)$ plane; an explanation for this feature of the
global bifurcation diagrams requires further study.

To summarise: our study was primarily motivated by the work of Honisch
\etal \cite{Honisch2015}. While we have clarified the mathematical
properties of (\ref{nl})--(\ref{nl2}) and (\ref{nln})--(\ref{nl2n}),
so that the structure of bifurcations in Figure 3(a) of that paper for
non-zero values of $\rho$ is now understood, many of their other
numerical findings are still to be explored mathematically.  For
example, the stability of ridge solutions shown in their Figure 5 in
the context of the full two-dimensional problem of a substrate with
periodic wettability stripes. There is clearly much work still to be
done on heterogeneous substrates with more complex wettability
geometry.

A final remark that might be of interest to the reader is that the
solutions of equations (\ref{nln})--(\ref{nl2n}), the steady state
  solutions of (\ref{qua}), a degenerate quasi-linear fourth-order
  PDE, can also be thought of as the steady state solutions of a much
  simpler (Rubinstein-Sternberg type \cite{RS}) second-order
  semi-linear non-local equation,
\begin{equation}\label{rs}
  h_t = \gamma h_{xx} +\Pi(h,x)-\frac{1}{L} \int_0^L \Pi(h,x)\, \de{x},
  \quad 0<x<L.
\end{equation}
It would be interesting to compare the dynamics of (\ref{qua}) and
(\ref{rs}), for example using the spectral comparison principles of
Bates and Fife \cite{BF}. For other work on non-local
reaction-diffusion equations such as (\ref{rs}), please see Budd \etal
\cite{Budd93} and the review of Freitas \cite{Freitas99}.

\begin{acknowledgement}
  We are grateful to Prof. U. Thiele (University of M\"unster) for
  clarifications concerning the work of Honisch \etal
  \cite{Honisch2015} and for sharing with us the AUTO codes used in
  that work which formed the basis of our continuation analysis. We
  are also grateful to the two anonymous referees whose remarks helped
  us to improve significantly the readability of the present work.
\end{acknowledgement}

\appendix

\section{$O(2)$ Symmetry Breaking by Spatial
  Non-homogeneity} \label{sb}

In this Appendix, we present an argument that shows that when the
wettability contrast is present, {\it i.e.} when $\rho \ne 0$, the
breaking of the $O(2)$ symmetry which equation (\ref{nln}) with the
periodic boundary conditions (\ref{nl2n}) has for $\rho=0$ , leaves
only two steady state solutions.

This is, in principle, a known result (see, for example, Chillingworth
\cite{Chillingworth1985}), but, since we are not aware of an easily
accessible reference, we give the details here. As before, we set
$G(x)=\sin (2\pi x)$. We provide the proof for $\Pisr$ given by
(\ref{DPsr}), the proof for $\Pilr$ given by (\ref{DPlr}) is
similar.

For the case of $\Pisr$, let us rewrite the boundary value problem
(\ref{nln}) in the form
\begin{equation}\label{app1}
\epsilon^2 h_{xx} + f_1(h) + \rho f_2(h)G(x) - \int_0^1 [f_1(h) +
\rho f_2(h)G(x)] \, \de{x} =0, \;\; 0< x < 1,
\end{equation}
where
\begin{equation}
  f_1(h)=\frac{1}{h^6} -\frac{1}{h^3},
\end{equation}
and
\begin{equation}
f_2(h) = \frac{1}{h^6},
\end{equation}
{\it i.e.} we separate the spatially homogeneous and spatially
non-homogeneous components of the disjoining pressure.  Equation
(\ref{app1}) is subject to the periodic boundary conditions
(\ref{nl2n}).

Suppose that when $\rho=0$ there is an orbit of steady state solutions,
{\it i.e.} a continuous closed curve of solutions $h_{0,s}(x)$,
parameterised by $s \in \R/[0,1]$, such that $h_{0,s}(x):=h_0(x+s)$,
for some function $h_0(x)$, i.e. all these solutions are related by
translation.  We aim to understand what remains of this orbit for
small non-zero $\rho$.

Fix $s \in \R/[0,1]$.  We write
\begin{equation}
h(x)= h_{0,s}(x)+ \rho h_1(x) + O(\rho^2).
\end{equation}
Substituting this expansion into (\ref{app1}) and collecting the
$O(\rho)$ terms, we have
\begin{eqnarray}\label{app2}
  \epsilon^2 h_{1,xx} & + & (f_1'(h_{0,s})+f_2'(h_{0,s}))h_1 - \int_0^1
  \left[f_1'(h_{0,s}) + f_2'(h_{0,s})\right] h_1 \, \de{x} \nonumber
  \\
  & = & - f_2(h_{0,s})G + \int_0^1 f_2(h_{0,s})G \, \de{x},
\end{eqnarray}
where, just like $h_{0,s}(x)$, $h_1(x)$ also satisfies
the periodic boundary conditions (\ref{nl2n}).

Now set
\begin{equation}
K u := \epsilon^2 u_{1,xx} + (f_1'(h_{0,s})+f_2'(h_{0,s}))u - \int_0^1 [f_1'(h_{0,s}) + f_2'(h_{0,s})] u \, \de{x},
\end{equation}
and let $D(K)$, the domain of $K$, be
\begin{equation}
  D(K) = \left\{f \in C^2\left(\left[0,1\right]\right) \; | \,
    f(0)=f(1), f'(0)=f'(1)\right\}.
\end{equation}
The operator $K$ is self-adjoint with respect to the $L^2([0,1])$
inner product. Invoking the Fredholm Alternative \cite[Theorem
7.26]{RY2008}, we conclude that (\ref{app2}) has $1$-periodic
solutions if and only if the right-hand side of (\ref{app2}) is
orthogonal in $L^2([0,1])$ to the solutions of $Ku=0$.

\medskip 

Next, we show that $u:= h_{0,s}'$ solves $Ku=0$. Indeed, by
differentiating (\ref{app1}) with $\rho=0$ with respect to $x$, we
find that $u$ solves the equation
\begin{equation}
  \epsilon^2 u_{xx}+ (f_1'(h_{0,s}) + f_2'(h_{0,s}))u=0.
\end{equation}
Integrating this equation over the interval $[0,1]$, we have that
\begin{equation}
\int_0^1 (f_1'(h_{0,s}) + f_2'(h_{0,s}))u\, \de{x} =0.
\end{equation}
Hence
\begin{equation}
\begin{split}
0 &~= \epsilon^2 u_{xx}+ (f_1'(h_{0,s}) + f_2'(h_{0,s}))u\\
  &~= \epsilon^2 u_{xx}+ (f_1'(h_{0,s}) + f_2'(h_{0,s}))u + \int_0^1
  (f_1'(h_{0,s}) + f_2'(h_{0,s}))u\, \de{x} \\
  &~= Ku.
\end{split}
\end{equation}

Also note that as $h_{0,s}(x)$ satisfies periodic boundary conditions, 
\begin{equation}\label{app3}
\int_0^1 h_{0,s}'(x) \de{x} = 0.
\end{equation}

Hence the solvability condition for (\ref{app2}) is
\begin{equation} \label{newapp} \int_0^1 h_{0,s}'(r)\left[
    -f_2(h_{0,s})G + \int_0^1 f_2(h_{0,s})G \, \de{x} \right] \de{r}
  =0.
\end{equation}
By (\ref{app3}), this condition reduces to   
\begin{equation}\label{app4}
\int_0^1 f_2(h_{0,s})h_{0,s}' G \, \de{x} =0.
\end{equation}

Now recall that $h_{0,s}(x) = h_0(x+s)$, so if we write
$F(x+s) = f_2(h_0(x+s))h_0'(x+s)$, the function $F(\cdot)$ is
$1$-periodic in $x$ with zero mean.  Hence
\begin{equation}
F(z)= \sum_{k=1}^\infty \alpha_k \sin (2k\pi z) +\beta_k \cos (2k\pi z).
\end{equation}
Therefore for  $G(x)=\sin(2 \pi x)$, the solvability condition
for (\ref{app2}) becomes
\begin{equation}
\alpha_1 \sin (2k\pi s)-\beta_1 \cos(2\pi s)=0,
\end{equation}
which has two solutions $ s \in \R/[0,1]$, from which we conclude there
is a solution $h_1(x)$ only for two choices of $s \in \R/[0,1]$, that
is, that only two solutions to (\ref{app1}) remain from the entire
$O(2)$ orbit that exists for $\rho=0$.

\end{document}